\documentclass[a4paper,11pt]{article}
\usepackage{fullpage}
\usepackage{amsmath, amssymb,amsfonts,amsthm}
\newtheorem{theorem}{Theorem}
\newtheorem{lemma}{Lemma}
\newtheorem{corollary}{Corollary}
\newtheorem{conjecture}{Conjecture}

\newtheorem{claim}{Claim}
\newcommand{\poly}{\textrm{poly}}
\newcommand{\eps}{\varepsilon}

\newcommand{\Real}{\mathbb{R}}
\newcommand{\Ran}{\mathcal{R}}
\newcommand{\aRan}{R}

\newcommand{\GF}{\textbf{GF}}
\renewcommand{\log}{\lg}

\newcommand{\cP}{\mathcal{P}}

\newcommand{\bits}{\{0,1\}}

\title{Adapt or Die: Polynomial Lower Bounds for Non-Adaptive Dynamic Data Structures}

\author{Joshua Brody $\quad$ Kasper Green Larsen \\\\CTIC\thanks{Center for the Theory of Interactive Computation, a Center of the Danish National Research Foundation.} $\textrm{ }$ and $\textrm{ }$ MADALGO\thanks{Center for Massive Data
    Algorithmics, a Center of the Danish National Research
    Foundation. }, \\ Department of Computer Science, \\ Aarhus
  University, Denmark\\ E-mail: \texttt{joshua.e.brody@gmail.com, larsen@cs.au.dk}} \date{}
\begin{document}
\maketitle
\begin{abstract}
In this paper, we study the role \emph{non-adaptivity} plays in
maintaining dynamic data structures. Roughly speaking, a data
structure is non-adaptive if the memory locations it reads and/or
writes when processing a query or update depend only on the query or
update and not on the contents of previously read cells.  We study
such non-adaptive data structures in the cell probe model. This model
is one of the least restrictive lower bound models and in particular,
cell probe lower bounds apply to data structures developed in the
popular word-RAM model. Unfortunately, this generality comes at a high
cost: the highest lower bound proved for any data structure problem is
only polylogarithmic.  Our main result is to demonstrate that one can
in fact obtain polynomial cell probe lower bounds for non-adaptive data
structures.

To shed more light on the seemingly inherent polylogarithmic lower
bound barrier, we study several different notions of non-adaptivity and
identify key properties that must be dealt with if we are to prove
polynomial lower bounds without restrictions on the data structures.

Finally, our results also unveil an interesting connection between
data structures and depth-2 circuits.  This allows us to translate
conjectured hard data structure problems into good candidates for high
circuit lower bounds; in particular, in the area of linear circuits
for linear operators.  Building on lower bound proofs for data
structures in slightly more restrictive models, we also present a
number of properties of linear operators which we believe are worth
investigating in the realm of circuit lower bounds.
\end{abstract}

\thispagestyle{empty}
\newpage
\setcounter{page}{1}
\section{Introduction}
\label{sec:intro} 
Proving lower bounds on the performance of data structures has been an
important line of research for decades. Over time, numerous
computational models have been proposed, of which the \emph{cell probe
  model} of Yao~\cite{yao:cellprobe} is the least restrictive. Lower
bounds proved in this model apply to essentially any imaginable data
structure, including those developed in the most popular upper bound
model, the word-RAM. Much effort has therefore been spent on deriving
cell probe lower bounds for natural data structure problems.
Nevertheless, the highest lower bound that has been proved for
\emph{any} data structure problem remains just polylogarithmic.

In this paper, we consider a natural restriction of data structures,
namely \emph{non-adaptivity}. Roughly speaking, a non-adaptive data
structure is a data structure for which the memory locations read
when answering a query or processing an update depend only on the
query or update itself, and not on the contents of the previously read
memory locations. Surprisingly, we are able to derive polynomially
high cell probe lower bounds for such data structures.

\subsection{The Cell Probe Model}
In the cell probe model, a data structure consists of a collection of
memory cells, each storing $w$ bits. Each cell has an integer address
amongst $[2^w]=\{1,\dots,2^w\}$, i.e. we assume any cell has enough
bits to address any other cell. When a data structure is presented
with a query, the query algorithm starts reading, or \emph{probing},
cells of the memory. The cell probed at each step may depend
arbitrarily on the query and the contents of all cells probed so
far. After probing a number of cells, the query algorithm terminates
with the answer to the query.

A dynamic data structure in the cell probe model must also support
updates. When presented with an update, the update algorithm similarly
starts reading and/or writing cells of the data structures. We refer
jointly to reading or writing a cell as probing the cell. The cell
probed at each step, and the contents written to a cell at each step,
may again depend arbitrarily on the update operation and the cells
probed so far.

The query and update times of a cell probe data structure are defined as the number of cells probed when answering a query or update respectively.
%The query time of a cell probe data structure is defined as the number
%of cells probed when answering a query, and the update time is
%similarly defined as the number of cells probed when processing an
%update. 
The space usage is simply defined as the largest address used by
any cell of the data structure.

\subsection{Previous Cell Probe Lower Bound Techniques}
\label{sec:previous}
As mentioned, the state-of-the-art techniques for proving cell probe
lower bounds unfortunately yield just polylogarithmic bounds. In the
following, we give a brief overview of the highest lower bounds that
has been achieved since the introduction of the model, and also the
most promising line of attack towards polynomial lower bounds.

\paragraph{Static Data Structures.}
One of the most important early papers on cell probe lower bounds for
static data structures is the paper of Miltersen et
al.~\cite{milt:asym}. They demonstrated an elegant reduction to data
structures from an assymmetric communication game. This connection
allowed them to obtain lower bounds of the form $t_q = \Omega(\lg
m/\lg S)$, where $m$ denotes the number of queries to the data
structure problem, $S$ the space usage in number of cells and $t_q$
the query time. Note however that this bound is insensitive to
polynomial changes in $S$ and cannot give super-constant lower bounds
for problems where the number of possible queries is just polynomial
in the input size (which is true for most natural problems).  This
barrier was overcome in the seminal work of P{\v a}tra{\c s}cu and
Thorup~\cite{patrascu10higher}, who extended the communication game of
Miltersen et al.~\cite{milt:asym} and obtained lower bounds of $t_q =
\Omega(\lg m/\lg(St_q/n))$, which peaks at $t_q = \Omega(\lg m/\lg \lg
m)$ for data structures using $n \poly(\lg m)$ space.

An alternative approach to static lower bounds was given by
Panigrahy et al.~\cite{pani:metric}. Their method is based on
sampling the cells of a data structure and showing that many queries
can be answered from a small set of cells if the query time is too
small (we note that similar ideas have been used for succinct data
structure lower bounds, see e.g.~\cite{gal:succinct}). The maximum lower bounds
that can be obtained from this technique are of the form $t_q =
\Omega(\lg m/\lg(S/n))$, see~\cite{larsen:staticloga}. For linear
space, this reaches $t_q = \Omega(\lg m)$, which remains the highest
static lower bound to date.

\paragraph{Dynamic Data Structures.}
The first technique for proving lower bounds on dynamic data
structures was the \emph{chronogram technique} of Fredman and
Saks~\cite{Fredman:chrono}. This technique gives lower bounds of the
form $t_q = \Omega(\lg n/\log(wt_u))$ and plays a fundamental role in
all later techniques for proving dynamic data structure lower bounds.
P{\v a}tra{\c s}cu and Demaine~\cite{Patrascu:loga} extended the
technique of Fredman and Saks with their \emph{information transfer}
technique.  This extension allowed for lower bounds of
$\max\{t_q,t_u\} = \Omega(\lg n)$.
Very recently, Larsen~\cite{larsen:dynamic_count} combined the
chronogram technique of Fredman and Saks with the cell sampling method
of Panigrahy et al. to obtain a lower bound of $t_q=\Omega((\lg
n/\lg(wt_u))^2)$, which remains the highest lower bound achieved so
far.

\paragraph{Conditional Lower Bounds.}
Examining all of the above results, we observe that no lower bound has
yet exceeded $\max\{t_u,t_q\}=\Omega((\lg n/\lg \lg n)^2)$ in the most
natural case of polynomially many queries, i.e. $m=\poly(n)$. In an
attempt to overcome this barrier, P{\v a}tra{\c
  s}cu~\cite{patrascu10mp-3sum} defined a dynamic version of a set
disjointness problem, named the \emph{multiphase problem}. We study
problems that are closely related to the multiphase problem, so we
summarize it here: 
\paragraph{The Multiphase Problem.}
This problem consists of three phases:
\begin{itemize}
  \item \textbf{Phase I:} In this phase, we receive $k$ sets
    $S_1,\dots,S_k$, all subset of a universe $[n]$. We are allowed to
    preprocess these sets into a data structure using time $O(\tau k
    n)$.
\item \textbf{Phase II:} We receive another set $T \subseteq [n]$ and
  have time $O(\tau n)$ to read and update cells of the data structure
  constructed in Phase I.
\item \textbf{Phase III:} We receive an index $i \in [k]$ and have
  time $O(\tau)$ to read cells of the data structure constructed
  during Phase I and II in order to determine whether $S_i \cap T =
  \emptyset$.
\end{itemize}

P{\v a}tra{\c s}cu conjectured that there exists constants $\mu > 1$
and $\eps>0$ such that any solution for the multiphase problem must
have $\tau = \Omega(n^\eps)$ when $k = n^\mu$, i.e. for the right
relationship between $n$ and $k$, any data structure must have either
polynomial preprocessing time, update time or query time. Furthermore,
he reduced the multiphase problem to a number of natural data
structure problems, including e.g. the following
problems. %(\cite{patrascu10mp-3sum} presents four additional problems
%that we have omitted here for brevity. We note that we obtain lower
%bounds for all of them.)
\begin{itemize}
\item\textbf{Reachability in Directed Graphs.} In a
preprocessing phase, we are given a directed graph with $n$ nodes and
$m$ edges. We are then to support inserting directed edges into the
graph. A query is finally specified by two nodes of the graph, $u$ and
$v$, and the goal is to determine whether there exists a directed path
from $u$ to $v$.

\item\textbf{Subgraph Connectivity.} In a preprocessing phase,
we are given an undirected graph with $n$ nodes and $m$ edges. We are
then to turn nodes on and off. A query is finally specified by two
nodes of the graph, $u$ and $v$, and the goal is to determine whether
there exists a path from $u$ to $v$ using only \emph{on} nodes.
\end{itemize}
We also mention the following problem, which was shown in~\cite{chan:mode} to
solve the multiphase problem.
\begin{itemize}
\item\textbf{Range Mode.} In a preprocessing phase, we are given an
  array $A[1 : n]=\{A[1],\dots,A[n]\}$ of integers and are to support
  value updates $A[i] \gets A[i] + x$.  Queries are specified by two indicies
  $i$ and $j$, and the goal is to find the most frequently occuring
  integer in the subarray $A[i:j]$.

\end{itemize}
These reductions imply polynomial lower bounds for the above problems,
if the multiphase problem has a polynomial lower bound. Thus it seems
fair to say that studying the multiphase problem is the most
promising direction for obtaining polynomial data structure lower
bounds.

\subsection{Non-Adaptivity}
\label{sec:nonadapt}
Given that we are generally clueless about how to prove polynomial
lower bounds in the cell probe model, it is natural to investigate
under which circumstances such bounds can be achieved. In this paper,
we study the performance of data structures that are non-adaptive. To
make the notion of non-adaptivity precise, we define it in the
following:

\begin{itemize}
\item \textbf{Non-Adaptive Query Algorithm.}
A cell probe data structure has a non-adaptive query algorithm, if the
cells it probes when answering a query depend only on the query, and
not on the contents of previously probed cells.

\item \textbf{Non-Adaptive Update Algorithm.}
Similarly, a cell probe data structure has a non-adaptive update
algorithm, if the cells it probes when processing an update depend
only on the update, and not on the contents of previously probed
cells.

\item \textbf{Memoryless Update Algorithm.}
In this paper, we also study a slighlty more restrictive type of
update algorithm.   A cell probe data structure has a memoryless update
algorithm, if the update algorithm is both non-adaptive, and
furthermore, the contents written to a cell during an update depend
only on the update and the current contents of the cell, i.e., they may
not depend on the contents of other cells probed during the update
operation.\footnote{A caveat on the semantics of updates: in this work, we
  assume updates specify how data \emph{changes} (e.g. updates are of
  the form $A[k] \leftarrow A[k] + \Delta$) as opposed to specifying
  new values for data (e.g. updates of the form $A[k] \leftarrow v$).
  The latter notion goes against the notion of non-adaptive updates,
  since to rewrite a cell, one must know how an update changes data.
  One solution is to assume that the data structure stores raw data
  directly, and to allow memoryless updates to depend on the current
  contents of a cell, the update, and the previous value of the
  update.  We view this issue as largely semantic, and do not discuss
  it further.}

\item \textbf{Linear Data Structures.}  
Finally, we study a sub-class of the data structures with a memoryless
update algorithm, which we refer to as linear data structures. These
data structures are defined for problems where the input can be
interpreted as an array $A$ of $n$ bits and an update operation can be
interpreted as flipping a bit of $A$ (from $0$ to $1$ or $1$ to
$0$). A linear data structure has non-adaptive query and update
algorithms. Furthermore, when processing an update, the contents of all
probed cells are simply flipped, and on a query, the data structure
returns the XOR of the bits stored in all the probed cells. Note that
these data structures use only a word size of $w=1$ bit, every cell
stores a linear combination over the bits of $A$ (mod 2) and a query
again computes a linear combination over the stored linear
combinations (mod 2).
\end{itemize}

While linear data structures might appear to be severly restrictive,
for many data structure problems (particularly in the area of range
searching), natural solutions are in fact linear. An example is the
well-studied \emph{prefix sum problem}, where the goal is to
dynamically maintain an array $A$ of bits under flip operations, and a
query asks for the XOR of elements in a prefix range $A[1\ldots k]$.
One-dimensional range trees are linear data structures that solve
prefix sum with update and query time $O(\log n)$. This is optimal
when memory cells store only single bits~\cite{Patrascu:loga}, even
for adaptive data structures. More elaborate problems in range
searching would be: Given a fixed set $P$ of $n$ points in
$d$-dimensional space, support deleting and re-inserting points of $P$
while answering queries of the form ``what is the parity of the number
of points inside a given query range?''. Here query ranges could be
axis-aligned rectangles, halfspaces, simplices etc. We note that all
the known data structures for range counting can easily be modified to
yield linear data structures when given a fixed set of points $P$, and
still, this setting seems to capture the hardness of range counting.

The main difference between non-adaptive and memoryless update
algorithms is that non-adaptive update algorithms may \emph{move} the
information about an update operation around the data structure, even
on later updates. As an example, consider a data structure with a
non-adaptive update algorithm and two possible updates, say updates
$u_1$ and $u_2$. Even if the data structure only probes the first
memory cell on update $u_1$, information about $u_1$ can be stored
many other places in the data structure.  Imagine the data structure
initially stores the value $0$ in the first memory cell. Whenever
update $u_1$ is performed, the data structure increments the contents
of the first memory cell by one. On update $u_2$, the data structure
copies the contents of the first memory cell to the second memory
cell. Clearly both operations are non-adaptive, and we observe that
whenever we have performed update $u_2$, the second memory cell stores
the number of times update $u_1$ has been performed, even though $u_1$
never probes the cell. For memoryless updates, information about an
update is only stored in cells that are actually probed when
processing the update operation.

Linear data structures are inherently memoryless.  However, some
features possible with memoryless updates are not available to linear
data structures.  For example, memoryless update algorithms can
support cells that maintain a count of the total number of updates
executed.  This is not possible with linear data structures, since the
contents of each cell is a \emph{fixed} linear combination of the data
being stored.

\subsection{Our Results}
\label{sec:our}
The main result of this paper, is to demonstrate that polynomial cell
probe lower bounds can be achieved when we restrict data structures to
be non-adaptive. 
In Section~\ref{sec:lower} we also prove lower bounds for data
structures where only the query algorithm is non-adaptive.  The
concrete data structure problem that we study in this setting is the
following indexing problem.
\paragraph{Indexing Problem.}  
In a preprocessing phase, we receive a set of $k$ binary strings
$S_1,\dots,S_k$, each of length $n$. We are then to support updates,
consisting of an index $j \in [n]$, which we think of as an index into
the strings $S_1,\dots,S_k$. A query is finally specified by an index
$i \in [k]$ and the goal is to return the $j$'th bit of $S_i$.
\begin{theorem}
\label{thm:nonquery}
Any cell probe data structure solving the indexing problem with a
non-adaptive query algorithm must either have $t_q = \Omega(n/w)$ or
$t_u = \Omega(k/w)$, regardless of the preprocessing time and space
usage.
\end{theorem}

Examining this problem, one quickly observes that it is a special case
of the multiphase problem presented in Section~\ref{sec:previous}, thus by
setting the parameters in the reductions
of~\cite{patrascu10mp-3sum,chan:mode} correctly we obtain, amongst
others, the following lower bounds as an immediate corollary of our
lower bound for the indexing problem:

\begin{corollary}
  Any cell probe data structure that uses a non-adaptive query
  algorithm to solve (i) reachability in directed graphs or (ii)
  subgraph connectivity must either have $t_q = \Omega(n/w)$ or $t_u =
  \Omega(n/w)$.  Any cell probe data structure that solves range mode
  with a non-adaptive query algorithm must have $t_qt_u =
  \Omega(n/w^2)$.
\end{corollary}

In Section~\ref{sec:lower}, we prove lower bounds for data
structures where the query algorithm is allowed to be adaptive, but
the update algorithm is memoryless. Again, we prove our lower bound
for a special case of the multiphase problem:
\paragraph{Set Disjointness Problem.}
In a preprocessing phase, we receive a subset $S$ of a universe
$[n]$. We are then to support inserting elements $x \in [n]$ into an
initially empty set $T$. Finally a query simply asks to return whether
$S \cap T = \emptyset$, i.e. the problem has just one query.
\begin{theorem}
\label{thm:memoryless}
Any cell probe data structure solving the set disjointness problem
with a memoryless update algorithm must have $t_q = \Omega(n/w)$,
regardless of the preprocessing time, space usage and update
time.
\end{theorem}
Again, using the reductions of~\cite{patrascu10mp-3sum,chan:mode}, we
obtain the following lower bounds as a corollary of our lower bound
for the set disjointness problem:
\begin{corollary}
  Any cell probe data structure that uses a memoryless update
  algorithm to solve (i) reachability in directed graphs, (ii)
  subgraph connectivity, or (iii) range mode must have $t_q =
  \Omega(n/w)$.
\end{corollary}
Finally, in Section~\ref{sec:circuit}, we show a strong connection
between nonadaptive data structures and the wire complexity of depth-2
circuits.  In these circuits, gates have \emph{unbounded} fan-in and
fan-out and compute \emph{arbitrary} functions.  Thus, trivial bounds
on the number of gates exist.  Instead, the size of a circuit $s(C)$ is
defined to be the number of wires.

Proving lower bounds on the size of circuits computing explicit
operators $F:\{0,1\}^n \rightarrow \{0,1\}^m$ has been studied in
several works.  In particular, Valiant~\cite{Valiant77} showed that an
$\omega(n^2/(\log\log n))$ bound for circuits computing $F$ implies
that $F$ cannot be computed by log-depth, linear size, bounded fan-in
circuits.  Currently, the best bounds known for an explicit operator
are $\Omega(n^{3/2})$.  Cherukhin~\cite{Cherukhin05} gave such a bound
for circuits computing cyclic convolutions.  Jukna~\cite{Jukna10} gave
a similar lower bound for circuits computing matrix multiplication,
and developed a general technique for proving such lower bounds,
formalizing the intuition in~\cite{Cherukhin05}.

First, we show how to use simple encoding arguments common to data
structure lower bounds to achieve circuit lower bounds, using matrix
multiplication as an example.  Our bound matches the result
from~\cite{Jukna10}, but yields a simpler argument.  We discuss
Jukna's technique in more detail in Section~\ref{sec:circuit}.

\begin{theorem}[\cite{Jukna10}]\label{thm:mm-encoding-intro}
  Any circuit computing matrix multiplication has size at least
  $n^{3/2}$.
\end{theorem}

Depth-2 circuits computing explicit \emph{linear} operators are of
particular interest.  Currently, the best lower bound for an explicit
linear operator is the recent $\Theta(n (\log n/\log \log n)^2)$ bound
of G\'{a}l et al.~\cite{Gal12} for circuits that compute error
correcting codes.  Another interesting question is whether general
circuits are more powerful than \emph{linear} circuits for computing
linear operators.  Linear circuits use only XOR gates; i.e., each gate
outputs a linear combination in $\GF(2)$ over its inputs.

We show a generic connection between linear data structures and linear
circuits.  Define a problem $\cP$ as a mapping $F_\cP =
(f_1,\dots,f_m) : \{0,1\}^n \to \{0,1\}^m$, where each $f_j :
\{0,1\}^n \to \{0,1\}$. For linear data structures, think of the
domain $\{0,1\}^n$ as the input array $A$ with $n$ bits, and view each
$f_j$ as a query, where $f_j(A)$ is the answer to the query $f_j$ on
the input $A$. A linear data structure hence solves $\cP$, if after
any sequence of updates to $A$, it holds for all $1 \leq j \leq m$
that answering the query $f_j$ returns the value $f_j(A)$.

\begin{lemma}\label{lem:circuit-intro}
  If there is a linear data structure for a problem $\cP$ with query
  time of $t_q$ and update time $t_u$, then there exists a depth-2
  linear circuit $C$ computing $F_\cP$ with size $s(C) \leq n t_u + mt_q$.

  If there is a depth-2 linear circuit $C$ that computes $F_\cP$, then there
  is a linear data structure for $\cP$ with \emph{average} query time
  at most $s(C)/m$ and average update time at most $s(C)/n$.
\end{lemma}

Lemma~\ref{lem:circuit-intro} thus gives a new way to attack circuit
lower bounds.  We believe the connection between non-adaptive data
structures and depth-2 circuits has the potential to yield strong
insight to this problem, and that several linear operators conjectured
to have strong data structure lower bounds are good candidates for
hard circuit problems (for linear or general circuits).

Apart from being interesting lower bounds in their own right, we
believe our results shed much light on the inherent difficulties of
proving polynomial lower bounds in the cell probe model. In particular
the \emph{movement} of data when performing updates (see the
discussion in Section~\ref{sec:nonadapt}) appears to be a major
obstacle.  We conclude in Section~\ref{sec:concl} with a discussion of
our results and potential directions for future research.
%In the following section, we move on to prove our lower
%bounds, and finally we conclude in Section~\ref{sec:concl} with a
%discussion of our results and potential directions for future
%research.

\section{Lower Bounds} \label{sec:lower}
%
%\subsection{Non-Adaptive Query Algorithm}
%\label{sec:query}
In this section, we first prove lower bounds for data structures where only
the query algorithm is assumed non-adaptive. The problem we study is
the indexing problem defined in Section~\ref{sec:our}.

\begin{theorem}
[Restatement of Theorem~\ref{thm:nonquery}]
Any cell probe data structure solving the indexing problem with a
non-adaptive query algorithm must either have $t_q = \Omega(n/w)$ or
$t_u = \Omega(k/w)$, regardless of the preprocessing time and space
usage. Here $t_q$ denotes the query time, $t_u$ the update time and $w$
the cell size in bits.
\end{theorem}

We prove this using an encoding argument. Specifically, consider
a game between an encoder and a decoder. The encoder receives as input
$k$ binary string $S_1,\dots,S_k$, each of length $n$ and must from
this send a message to the decoder. From the message alone, the
decoder must uniquely recover all the strings $S_1,\dots,S_k$. If the
strings $S_1,\dots,S_k$ are drawn from a distribution, then the
expected length of the message must be at least $H(S_1 \cdots S_k)$,
or we have reached a contradiction.  Here $H(\cdot)$ denotes Shannon
entropy.

The idea in our proof is to assume for contradiction that a data
structure for the indexing problem exists with a non-adaptive query
algorithm that simultaneously has $t_q=o(n/w)$ and $t_u=o(k/w)$. Using
this data structure as a black box, we construct a message that is
shorter than $H(S_1 \cdots S_k)$, but at the same time, the decoder
can recover $S_1,\dots,S_k$ from the message, i.e. we have reached the
contradiction.  We let the $k$ strings $S_1,\dots,S_k$ given as input
to the encoder be uniform random bit strings of length $n$. Clearly
$H(S_1 \cdots S_k) = kn$.
\paragraph{Encoding Procedure.}
When given the strings $S_1,\dots,S_k$ as input, the encoder first
runs the preprocessing algorithm of the claimed data structure on
$S_1,\dots,S_k$. He then examines every possible query index $i \in
[k]$, and for each $i$, collects the set of addresses of the cells
probed on query $i$. Since the query algorithm is non-adaptive, these
sets of addresses are independent of $S_1,\dots,S_k$ and any updates
we might perform on the data structure. Letting $C$ denote the set
containing all these addresses for all $i$, the encoder starts by
writing down the concatenation of the contents of all cells with an
address in $C$. This constitutes the first part of the message.% to the decoder.

The encoder now runs through every possible update $j\in [n]$. For
each $j$, he runs the update algorithm as if update $j$ was
performed on the data structure. While running update $j$, the decoder
appends the contents of the probed cells (as they are when the update
reads the cells, not after potential changes) to the constructed
message. After processing all $j$'s, the encoder finally sends
the constructed message to the decoder. This completes the encoding
procedure.

\paragraph{Decoding Procedure.}
The decoder receives as input the message consisting first of the
contents of all cells with an address in $C$ after preprocessing
$S_1,\dots,S_k$. Since the query algorithm is non-adaptive, the
decoder knows the addresses of all these cells simply by examining the
query algorithm of the claimed data structure. The decoder will now
run the update algorithm of every $j \in [n]$. While doing this, he
maintains the contents of all cells in $C$ and all cells probed during
the updates. Specifically, the decoder does the following:

For each $j=1,\dots,n$ in turn, he starts to run the update algorithm
for $j$. Observe that the contents of each probed cell (before
potential changes) can be recovered from the message (the contents
appear one after another in the message). This allows the decoder to
completely simulate the update algorithm for each $j=1,\dots,n$. Note
furthermore that for each cell that is probed during these updates,
the address can also be recovered simply by examining the update
algorithm. In this way, the decoder always knows the contents of all
cells in $C$ and all cells probed by the update algorithm as they
would have been after preprocessing $S_1,\dots,S_k$ and performing the
updates after this preprocessing. While processing the updates
$j=1,\dots,n$, the decoder also executes a number of queries: After
having completely processed an update $j$, the decoder runs the query
algorithm for every $i \in [k]$. Note that the decoder knows the
contents of all the probed cells as if the preprocessing on
$S_1,\dots,S_k$ had been performed, followed by updates
$j'=1,\dots,j$. This implies that the simulation of the query
algorithm for each $i \in [k]$ terminates precisely with the answer
being the $j$'th bit of $S_i$. It follows immediately that the decoder
can recover every bit of every $S_i$ from the message.

\paragraph{Analysis.}
What remains is to analyze the size of the message. Since by
assumption, the query time is $t_q=o(n/w)$, the first part of the
message has $t_qkw=o(kn)$ bits. Similarly, we assumed $t_u=o(k/w)$,
thus the second part of the message has $t_unw=o(kn)$ bits. Thus the
entire message has $o(kn)$ bits. Since $H(S_1 \cdots S_k) = kn$, we
have reached our contradiction. This completes the proof of
Theorem~\ref{thm:nonquery}.\\\\
%\subsection{Memoryless Update Algorithm}
%\label{sec:update}
%In the following
Next, we prove lower bounds for data structures where only
the update algorithm is assumed to be memoryless, that is, we allow
the query algorithm to be adaptive. In this setting, we study the
set disjointness problem defined in Section~\ref{sec:our}:

\begin{theorem}
[Restatement of Theorem ~\ref{thm:memoryless}]
Any cell probe data structure solving the set disjointness problem
with a memoryless update algorithm must have $t_q = \Omega(n/w)$,
regardless of the preprocessing time, space usage and update
time. Here $t_q$ denotes the query time and $w$ the cell size in bits.
\end{theorem}
%
%\paragraph{Alternate Proof}
%Joshua: below is an alternate writeup of this encoding argument, which
%hopefully improves on the readability.  (I am not 100\% committed to
%this version.  we can discuss)

Again, we prove this using an encoding argument. In this encoding
proof, we let the input of the encoder be a uniform random set $S
\subseteq [n]$. Clearly $H(S)=n$ bits. We now assume for contradiction
that there exists a data structure for the set disjointness problem
with a memoryless update algorithm and at the same it has query time
$t_q = o(n/w)$. The encoder uses this data structure to send a message
encoding $S$ in less than $n$ bits, i.e. a contradiction.

\paragraph{Encoding Procedure.}
When the encoder receives $S$, he runs the preprocessing algorithm of
the claimed data strucutre. Then, he computes $\bar{S} = [n]\setminus
S$ and inserts $\bar{S}$ into the data structure as the set $T$.
Finally, the encoder runs the query algorithm and notes the set of
cells $C$ probed.  Note that by the choice of $\bar{S}$, the query
algorithm will output \emph{disjoint}, and furthermore, $\bar{S}$ is
the largest possible set that will result in a \emph{disjoint}
answer.

The encoding consists of three parts\footnote{In fact, it is possible for the decoder to recover
  $C$ from the second two parts of the encoding, so the first part is
  unnecessary.  However, this does not materially affect our lower
  bound, so we omit the details.}: (i) the addresses of the cells in $C$, (ii)
the contents of the cells in $C$ after preprocessing but before inserting
$\bar{S}$, and (iii) the contents of the cells in $C$ after inserting
$\bar{S}$.

\paragraph{Decoding Procedure.}
The decoder iterates over all sets $S' \subseteq [n]$.  Each time, the
decoder initializes the contents of cells in $C$ to match the second
part of the encoder's message.  Then, he inserts each element of $S'$
into the data structure, changing the contents of any cell in $C$
where appropriate.  When a cell outside of $C$ is to be changed, the
decoder does nothing.  Since the update algorithm is memoryless, this
procedure ends with all cells in $C$ having the same contents as they
would have had after preprocessing $S$ and inserting elements of $S'$.
Moreover, if the contents match the contents written down in the third
part of the encoding, then it must be that $S$ and $S'$ are disjoint
(we know that the query answers \emph{disjoint} when the contents of
$C$ are like that). When $S' = \bar{S}$, the contents of $C$ will
match the last part of the encoding, and it is trivially the largest
set to do so.  Thus, the decoder selects the largest set $S^*$ whose
updates to $C$ match the contents written down in the third part of
the encoding.  In this way, the decoder recovers $S = [n] \setminus
S^*$.

\paragraph{Analysis.}
Finally, we analyze the size of the encoding.  Since we assumed $t_q =
o(n/w)$, the encoding has size $3t_qw = o(n)$ bits.  But
$H(S) = n$, thus we have reached a contradiction.

\section{Circuits and Non-Adaptive Data Structures}\label{sec:circuit}
In this section, we demonstrate a strong connection between
non-adaptive data structures and the wire complexity of depth-2
circuits.
A depth-2 circuit computing $F = (f_1,\ldots, f_m) : \bits^n
\rightarrow \bits^m$ is a directed graph with three layers of
vertices.  The first layer consists of $n$ input nodes, labeled
$x_1,\ldots, x_n \in \bits$.  Vertices in the second layer are
interior gates and output boolean values.  The last layer consists of $m$
output gates, labeled $z_1,\ldots, z_m \in \bits$.  There are edges
between input nodes and interior gates, and between interior gates and
output gates.  Each gate computes an \emph{arbitrary} function of its
inputs.  Since non-input nodes compute arbitrary functions, $f$ can be
trivially computed using $m$ gates.  Instead, we define the
\emph{size} $s(C)$ of a depth-2 ciruit $C$ as the total number of
wires in it; i.e., the number of edges in the graph.

%\subsection{Circuit Lower Bounds via Encoding}
First, we show how to use the encoding technique common to data
structure lower bounds to achieve size bounds for depth-2 circuits.
As a proof of concept, we prove such a lower bound for matrix
multiplication.  We say that a circuit computes matrix multiplication
if there are $n = 2m$ inputs, each corresponding to an entry in one of
two $\sqrt{n} \times \sqrt{n}$ binary matrices $A$ and $B$, and each output
gate computes an entry in the product $A\cdot B$.  Arithmetic is in $\GF(2)$.

Jukna~\cite{Jukna10} considered depth-$2$ circuits and gave an
$n^{3/2}$ lower bound for circuits computing boolean matrix multiplication.
At a high level, his proof proceeds in the following fashion.
\begin{enumerate}
\item Partition input nodes into sets $I_1,\ldots, I_t$ and output gates
  into sets $J_1,\ldots, J_t$.
\item Prove that for each $1 \leq \ell \leq t$, the number of wires
  leaving inputs from $I_\ell$ plus the number of wires entering
  outputs in $J_\ell$ must be large.\label{foo:part}
\item Conclude a large lower bound by summing the terms from Step~\ref{foo:part}.
\end{enumerate}
Note that since $\{I_\ell\}$ and $\{J_\ell\}$ are partitions, they
induce a partition on the wires in the circuit.  Jukna proves
Step~\ref{foo:part} by proving lower bounds on what he calls the
\emph{entropy} of an operator.  He proves a lower bound on the entropy
of an operator by carefully analyzing subfunctions of the operator.
In the case of matrix multiplication, subfunctions are created by
fixing entries in $B$ to be all zero, except for a single cell
$B[k,\ell]$.  Each $I_\ell,J_\ell$ represents a column in $B$ and in
$A \cdot B$ respectively.  By ranging over different $k,\ell$, Jukna
is able to argue that the entropy of matrix multiplication is high.
The details of this argument are technical.

We give a new proof for Step~\ref{foo:part} using an
encoding argument.  The encoder exploits the circuit operations to
encode a $\sqrt{n} \times \sqrt{n}$ matrix $A$.  The encoded message
has length precisely equal to the nubmer of outgoing wires in $I_\ell$
and incoming wires to $J_\ell$.  The argument is very similar
to the arguments in Section~\ref{sec:lower}; we leave it to the full
version of the paper for lack of space.

\begin{theorem}\label{thm:mm-encoding}
  Any circuit $C$ computing boolean matrix multiplication has size $s(C)
  \geq n^{3/2}$.
\end{theorem}

%\subsection{Linear Circuits and Linear Data Structures}
%\label{sec:queryandupdate}
%In this section, 
Finally, we provide a strong connection between depth-2 linear
circuits and linear data structures. The connection is almost
immediately established:

\begin{lemma}[Restatement of Lemma~\ref{lem:circuit-intro}]\label{lem:circuit}
  If there is a linear data structure for a problem $\cP$ with query
  time of $t_q$ and update time $t_u$, then there exists a depth-2
  linear circuit $C$ computing $F_\cP$ with size $s(C) \leq n t_u +
  mt_q$.

  If there is a depth-2 linear circuit $C$ computing $F_\cP$, then there
  is a linear data structure for $\cP$ with \emph{average} query time
  at most $s(C)/m$ and average update time at most $s(C)/n$.
\end{lemma}
\begin{proof}
  First, suppose there exists a linear data structure solving $\cP$.
  We construct the corresponding depth-2 circuit directly.  Input
  nodes correspond to the $n$ bits of the input (the array $A$ in the
  definition of linear data structures). Output nodes correspond to
  the $m$ possible queries, and there is an interior node for each
  cell in the database. For each update $1\leq i \leq n$ (flip an
  entry of $A$), add edges from $x_i$ to each of the cells updated by
  the data structure. Similarly, add wires $(c_i,z_j)$ whenever the
  $j$th query probes the $i$th cell in the data structure. Correctness
  follows immediately. Finally, note that since updates and queries
  probe at most $t_u$ and $t_q$ cells respectively, the total number
  of wires in the circuit is bounded by $s(C) \leq nt_u + mt_q$.
  
  Constructing a linear data structure from a linear depth-2 circuit
  $C$ is similar. Letting $t_{u,i}$ and $t_{q,j}$ denote the number
  of cells probed during the $i$th update and $j$th query
  respectively, it is easy to see that $s(C) = \sum_{i=1}^n t_{u,i} +
  \sum_{j=1}^m t_{q,j}$.  It follows that the average update time is
  at most $\frac{1}{n}\sum t_{u,i} \leq s(C)/n$, and similarly that
  the average query time is at most $\frac{1}{m}\sum t_{q,j} \leq
  s(C)/m$.
\end{proof}

%currently leaving out, since current draft has subsection on different models.
%\paragraph{Remark.}  The reduction from circuits to nonadaptive data
%structures is stronger than we claim in this paper.  Specifically,
%when updating a cell in the data structure, we could allow the
%updating algorithm free read-only access to the set of updates that
%have probed the cell.  Circuit lower bounds would also imply cell
%probe lower bounds for this kind of data structure.\\

The main contribution of Lemma~\ref{lem:circuit} is a new range of
candidate hard problems for linear circuits, all inspired by data
structure problems. As mentioned in Section~\ref{sec:nonadapt}, linear
data structures most naturally occur in the field of range
searching. Furthermore, these data structure problems turn out to
correspond precisely to linear operators: Let $P=\{p_1,\dots,p_n\}$ be
a fixed set of $n$ points in $\Real^d,$ and let $\Ran$ be a set of
query ranges, where each $\aRan_i \in \Ran$ is a subset of
$\Real^d$. $P$ and $\Ran$ naturally define a linear operator
$A(P,\Ran) \in \{0,1\}^{|\Ran| \times |P|}$, where the $i$th row of
$A(P,\Ran)$ has a $1$ in the $j$th column if $p_j \in \aRan_i$ and $0$
otherwise. In the light of Lemma~\ref{lem:circuit}, assume a linear
data structure solves the following range counting problem: Given the
fixed set of points $P$, each assigned a weight in $\{0,1\}$, support
flipping the weights of the points (intuitively inserting/deleting the
points) while also supporting to efficiently compute the parity of the
weights assigned to the points inside a query range $\aRan_i \in
\Ran$. Then that linear data structure immediately translates into a
linear circuit for the linear operator $A(P,\Ran)$ and vice
versa. Thus we expect that hard range searching problems of the above
form also provide hard linear operators for linear circuits. The
seemingly hardest range searching problem is \emph{simplex range
  searching}, where we believe that the following holds:

\begin{conjecture}
\label{conj:simplex}
There exists a constant $\eps>0$, a set $\Ran$ of $\Theta(n)$
simplices in $\Real^d$ and a set of $n$ points in $\Real^d$, such that
any data structure solving the above range counting problem (flip
weights, parity queries), must have average query and update time $t_u
t_q = \Omega(n^\eps)$.
\end{conjecture}

We have toned down Conjecture~\ref{conj:simplex} somewhat, since the
community generally believe $\eps$ can be replaced by $1-1/d$, but to be
on the safe side we only conjecture the above. In the circuit setting,
this conjecture translates to
\begin{corollary}
If Conjecture~\ref{conj:simplex} is true for linear data structures,
then there exists a constant $\delta>0$, a set $\Ran$ of $\Theta(n)$
simplices in $\Real^d$ and a set $P$ of $n$ points, such that any
linear circuit computing the linear operator $A(P,\Ran)$ must have
$\Omega(n^{1+\delta})$ wires.
\end{corollary}

Furthermore, the research on data structure lower bounds also provide
a lot of insight into which concrete sets $P$ and $\Ran$ that might be
difficult. More specifically, polynomial lower bounds for simplex
range searching has been proved for: range reporting in the pointer
machine~\cite{chazelle:simplex,Afshani.simplex} and
I/O-model~\cite{Afshani.simplex}, range searching in the semi-group
model~\cite{chazelle:polytope} and range searching in the group
model~\cite{larsen:group,larsen:improved}. The group model comes
closest in spirit to linear data structures. A data structure in the
group model is essentially a linear data structure, where instead of
storing linear combinations over $\GF(2)$, we store linear
combinations with integer coefficients (and no mod
operations). Similarly, queries are answered by computing linear
combinations over the stored elements, but with integer coefficients
and not over $\GF(2)$. The properties used to drive home range
searching lower bounds in the group model are:
\begin{itemize}
\item If $A(P,\Ran)$ has polynomial red-blue discrepancy, then any
  group model data structure must have $t_u t_q = \Omega(n^\eps)$ for some
  constant $\eps>0$.
\item If $A(P,\Ran)$ has $\Omega(n)$ eigenvalues that are polynomial,
  then any group model data structure must have $t_u t_q =
  \Omega(n^\eps)$ for some constant $\eps>0$.
\item If $|\aRan_i \cap P|$ is polynomial for all $\aRan_i \in \Ran$
  and $|\aRan_i \cap \aRan_j \cap P|=O(1)$ for all $i \neq j$, then
  any group model data structure must have $t_u t_q = \Omega(n^\eps)$
  for some constant $\eps>0$.
\end{itemize}
The last property directly translates to $A(P,\Ran)$ having rows and
columns with polynomially many $1$s and any two rows/columns having a
constant number of $1$s in common. Given the tight correspondence
between group model data structures and linear data structures, we
believe these properties are worth investigating in the circuit
setting.  Furthermore, a concrete set of $n$ points $P$ and a set of
$\Theta(n)$ simplices $\Ran$, with all three properties, is known even
in $\Real^2$. This example can be found in~\cite{chazelle:fracLB},
where it is stated for $\Ran$ being lines (i.e. degenerate
simplices). Note that the lower bound in~\cite{chazelle:fracLB} is for
range reporting in the pointer machine, but using the observations
in~\cite{larsen:group,larsen:improved} it is easily seen that all the
above properties hold.

Even if these properties are not enough to obtain lower bounds for
linear operators, we believe the geometric approach might be useful in
its own right.

\section{Conclusion}
\label{sec:concl}
In this paper, we have studied the role non-adaptivity plays in
dynamic data structures. Surprisingly, we were able to prove
polynomially high lower bounds for such data structures. Perhaps more
importantly, we believe our results shed much new light on the
current polylogarithmic barriers if we do not make any restrictions on
data structures. We also presented an interesting connection between
data structures and depth-2 circuits. The connection between linear
operators and range searching is particularly intriguing, revealing a
number of new properties to investigate further in the realm of
circuit lower bounds.

\section{Acknowledgements}
We are grateful to Elad Verbin for several helpful discussions.  
%We
%dedicate this paper to the memory of Mihai P\v{a}tra\c{s}cu, a giant
%on whose shoulders we stand.
%\newpage

\appendix
\section{A Lower Bound Proof for Matrix Multiplication}
\begin{theorem}[Restatement of Theorem~\ref{thm:mm-encoding}]
  Any circuit $C$ computing boolean matrix multiplication has size $s(C)
  \geq n^{3/2}$.
\end{theorem}
\begin{proof}
  Fix a circuit $C$.  Let $P = A\cdot B$.
  For $1 \leq \ell \leq \sqrt{n}$, let $I_\ell$ denote the $\ell$th
  column of $B$; that is, $I_\ell$ consists of all inputs
  corresponding to $B[k,\ell]$ for some $k$.  Similarly, $J_\ell$ is
  the set of all outputs corresponding to the $\ell$th column of $P$;
  that is, all outputs given by $P[k,\ell]$ for some $k$.
  Let $t_{u,\ell}$ denote the number of wires leaving inputs in
  $I_\ell$.  Similarly, let $t_{q,\ell}$ denote the number of wires
  entering outputs in $J_\ell$.
  \begin{claim}\label{claim:mm-encoding}
    For any $\ell$, we have $t_{u,\ell} + t_{q,\ell} \geq n$.
  \end{claim}
  Before proving this claim, note that Theorem~$\ref{thm:mm-encoding}$
  follows directly, since there are $\sqrt{n}$ pairs $(I_\ell,
  J_\ell)$ and the wires corresponding to each pair are disjoint.
\end{proof}
  \begin{proof}[Proof of Claim~\ref{claim:mm-encoding}]
    This proof will involve an encoding argument.  The encoder will
    receive a $\sqrt{n} \times \sqrt{n}$ boolean matrix $M$, where $M$
    is drawn uniformly amongst all such boolean matrices. He will then
    use the matrix multiplication circuit to encode $M$ in such a way
    that the size of the encoding depends on the wires leaving
    $I_\ell$ and entering $J_\ell$.
    \paragraph{Encoding Procedure.}
    The encoder receives $M$.  As a first step, he sets $A[i,j]
    \leftarrow M[i,j]$ for all $i,j$; he also sets all entries in $B$ to zero.  He
    then writes down the output of all interior gates adjacent
    to an output in $J_\ell$.  In the second step, for each $1 \leq k
    \leq \sqrt{n}$, the encoder performs the following: he sets $B[k,\ell]
    \leftarrow 1$ and sets all other entries in $B$ to zero.  He then
    writes down the output of all interior gates adjacent to
    $B[k,\ell]$.  This completes the encoding procedure.
    \paragraph{Decoding Procedure.}
    Note that $P[i,\ell] = \sum_j A[i,j]B[j,\ell]$.  In particular,
    when $B$ consists of a $1$ in entry $[k,\ell]$ and zero in all
    other entries, then the $\ell$th column of $P$ corresponds to the
    $k$th column of $A$.  The decoder thus recovers the $k$th column
    of $M$ by using $C$ to compute the $\ell$th column of $P$, i.e., by
    querying all outputs in $J_\ell$.  For each output gate in
    $J_\ell$, she looks at all interior gates adjacent to it.  For
    each of \emph{these} gates $g$, the decoder checks to see if $g$
    is adjacent to the input gate $B[k,\ell]$.  If so, then she
    recovers the correct output value of this gate from the second
    part of the encoding.  Otherwise, she recovers the correct output
    from the first part (noting that in this case, changing the value
    of $B[k,\ell]$ does not affect $g$).  In this way, the decoder
    recovers the $\ell$th column of $C$, which is also the $k$th
    column of $A$, which is again the $k$th column of $M$. Doing this
    for all $k$ completes the decoding.
    \paragraph{Analysis.}
    The first part of the encoding consists of the output of each
    interior gate adjacent to at least one output in $J_\ell$.  Thus,
    the first part of the encoding can be described in at most
    $t_{q,\ell}$ bits.  The second part of the encoding consists of
    the output of each interior gate adjacent to each input node in
    $I_\ell$.  This requires at most $t_{u,\ell}$ bits.  Thus,
    the total length of the encoding is at most $t_{u,\ell} +
    t_{q,\ell}$.  The decoder recovers all of $M$ from this message.
    Since each entry of $M$ is independent and uniform, $H(M) = n$.
    Thus, $t_{u,\ell} + t_{q,\ell} \geq n$.
  \end{proof}

\paragraph{Remark.}
  As mentioned previously, Jukna proves his lower bounds by defining
  the \emph{entropy} of an operator.  He lower bounds the wire
  complexity of a circuit by the entropy of the operator it computes.
  He proves a lower bound on the entropy of an operator by carefully
  analyzing subfunctions of the operator, created by fixing subsets of
  the variables to specific values and considering the induced
  function on the remaining variables.

  Parts of Jukna's proof are similar in spirit to ours.  In
  particular, the way we encode $M$ by fixing the matrix $B$ to be one
  in entry $[k,\ell]$ and zero elsewhere corresponds to the
  subfunctions Jukna considers in his proof.  In fact, we believe that
  \emph{any} lower bound provable using Jukna's technique can also be
  proved using our method.  Our advantage is in replacing Jukna's
  technical and somewhat complicated machinery with a simple encoding
  argument.

\end{document}